%% file: purification.tex
\pdfoutput=1
\documentclass[a4paper,onecolumn]{quantumarticle}
\usepackage[utf8]{inputenc}
\usepackage[T1]{fontenc}
\usepackage{lmodern}
\usepackage[UKenglish]{babel}
\usepackage[unicode,colorlinks]{hyperref}
\usepackage{graphicx}
\usepackage{amsmath,amssymb,amsthm}
\usepackage{xspace}
\usepackage{mathtools}
\usepackage{dsfont}
\usepackage{xcolor}
\usepackage{tikz}

\usepackage{csquotes}
\usepackage[backend=bibtex,style=links,maxnames=10,defernumbers=true,giveninits=true]{biblatex} 
\bibliography{biblio}

\input{comandos}

\hypersetup{pdftitle={{A purification postulate for quantum mechanics with indefinite causal order}}}

\begin{document}
\title{A purification postulate for quantum mechanics with indefinite causal order}
\author{Mateus Araújo}
\affiliation{Faculty of Physics, University of Vienna, Boltzmanngasse 5 1090 Vienna, Austria}
\affiliation{Institute for Quantum Optics and Quantum Information (IQOQI), Boltzmanngasse 3 1090 Vienna, Austria}
\affiliation{Institute for Theoretical Physics, University of Cologne, Germany}
\author{Adrien Feix}
\affiliation{Faculty of Physics, University of Vienna, Boltzmanngasse 5 1090 Vienna, Austria}
\affiliation{Institute for Quantum Optics and Quantum Information (IQOQI), Boltzmanngasse 3 1090 Vienna, Austria}
\author{Miguel Navascués}
\affiliation{Institute for Quantum Optics and Quantum Information (IQOQI), Boltzmanngasse 3 1090 Vienna, Austria}
\author{Časlav Brukner}
\affiliation{Faculty of Physics, University of Vienna, Boltzmanngasse 5 1090 Vienna, Austria}
\affiliation{Institute for Quantum Optics and Quantum Information (IQOQI), Boltzmanngasse 3 1090 Vienna, Austria}
\date{\today}

\begin{abstract}
To study which are the most general causal structures which are compatible with local quantum mechanics, Oreshkov \etal \cite{oreshkov12} introduced the notion of a process: a resource shared between some parties that allows for quantum communication between them without a predetermined causal order. These processes can be used to perform several tasks that are impossible in standard quantum mechanics: they allow for the violation of causal inequalities, and provide an advantage for computational and communication complexity. Nonetheless, no process that can be used to violate a causal inequality is known to be physically implementable. There is therefore considerable interest in determining which processes are physical and which are just mathematical artefacts of the framework. Here we make key progress in this direction by proposing a purification postulate: processes are physical only if they are purifiable. We derive necessary conditions for a process to be purifiable, and show that several known processes do not satisfy them.
\end{abstract}
\maketitle

\section{Introduction}

It is widely believed that any future theory that unifies quantum mechanics and gravity will feature quantum uncertainty in the metric tensor \cite{ashtekar96}, which should produce indefinite causal structures. Our understanding of the notion of indefinite causal structures is, however, still lacking. To investigate that, one approach is to consider simple, concrete models that are compatible with the laws of quantum information processing. One such model -- the process matrix formalism -- was introduced by Oreshkov \etal as the most general causal structure that can connect local laboratories where standard quantum mechanics is valid without creating paradoxical causal loops~\cite{oreshkov12}.

These processes have been shown to enable the realization of tasks that are otherwise impossible: they allow for the violation of causal inequalities \cite{oreshkov12,baumeler13,baumeler14,branciard15simplest,baumeler16,oreshkov15,abbott16}, can be detected by causal witnesses \cite{araujo15witnessing}, provide an advantage in quantum computation \cite{chiribella09,chiribella12,araujo14}, and enable a reduction in communication complexity \cite{feix15,allardguerin16}. But even though one such process -- the quantum switch -- has been implemented experimentally \cite{procopio_experimental_2014,rubino16}, in general it is not known if all process are physical or some of them are just mathematical artefacts of the formalism. They were, after all, defined only from the requirement of not generating logical contradictions, and processes realisable in nature are likely to fulfil additional physical constraints. 

One can, therefore, look for requirements beyond mere logical consistency in order to shed light on the physicality of these processes. As seen from the search for physical principles to determine the set of quantum correlations, such meta-theoretical principles can provide nontrivial constraints on the possible theories \cite{brassard06,linden07,pawlowski09,navascues09,fritz13}. The principle we choose to investigate here is reversibility of the transformations between quantum states \cite{chiribella10}: it is, after all, valid in all fundamental physical theories, and has been a central ingredient in all of the reconstructions of quantum mechanics to date \cite{hardy01,dakic09,masanes11,chiribella11,barnum14,hoehn14b,hoehn15}.

To define what reversibility means for processes we first needed to generalise their definition: we extend them with a global past and a global future, so that they can be seen as inducing a transformation from quantum states in the past to quantum states in the future, after passing through the causally indefinite region of the local laboratories. We can then define pure processes as those that preserve the reversibility of the underlying operations, \ie, those that induce a unitary transformation from the past to the future whenever unitary transformations are also applied in the local laboratories.

With these definitions in hand, we can propose the purification postulate: \textit{processes are physical only if they are purifiable}, \ie, if they can be expressed as a part of a pure process in a larger space. It turns out to be rather difficult to determine if a given process is purifiable. We derive, then, a necessary but not sufficient condition for purifiability, and show that several known processes fail to satisfy this condition, among which is the process from Ref.~\cite{oreshkov12} that was the first one shown to violate a causal inequality. In fact, we haven't found a purification for any bipartite process that was able to violate a causal inequality. There exists, however, a tripartite process which is purifiable and able to violate causal inequalities \cite{araujo14private,baumeler16,baumeler15private}. If one takes the view that the violation of causal inequalities is impossible in Nature, as done in Ref.~\cite{feix16}, this implies that being purifiable is not a sufficient condition for a process to be physical.

The paper is organized as follows: in Section \ref{sec:pure} we generalise the definition of a process and introduce the notion of pure processes. In Section \ref{sec:postulate} we propose the purification postulate. In Section \ref{sec:ifandonlyif} we derive necessary and sufficient conditions for a process to be purifiable, and in Section \ref{sec:onlyif} we derive a necessary but not sufficient condition for purifiability that can be easily tested. In Section \ref{sec:examples} we apply our condition to several known processes, and in Section \ref{sec:counterexample} we present a tripartite pure process that violates causal inequalities.

\section{Pure processes}\label{sec:pure}

In Ref.~\cite{oreshkov12}, the process matrix was introduced as the most general way to allow two parties Alice and Bob to communicate -- not necessarily in a causally ordered way -- but without creating paradoxes. These parties were assumed to obey the laws of quantum mechanics locally, and the no-paradox condition means that whatever probabilities Alice and Bob may extract from their local experiments will be positive and normalised. In other words, the process matrix was defined as a (multi)linear function that takes completely positive (CP) maps to probabilities.

To be able to talk about purification, we need to extend this definition to take CP maps not to probabilities, but to other CP maps. This view of a process as a higher-order transformation is much in the spirit of quantum combs \cite{chiribella09b}, but with the crucial difference that combs are usually defined to be causally ordered. Note also that it can be recovered from the multipartite definition of process presented for example in Ref.~\cite{araujo15witnessing}. We shall, nevertheless, explicitly define bipartite processes as higher-order transformations for clarity and to make this paper self-contained. The multipartite case can be obtained by a straightforward generalisation of the arguments presented here or from Ref.~\cite{araujo15witnessing}.

To this end, let $\mathcal A$ be a completely positive trace preserving (CPTP) map  that takes the Hilbert spaces\footnote{$A_I$ (and its analogues) is defined as the set of $d_{A_I} \times d_{A_I}$ complex matrices.} $A_I,A_I'$ to $A_O,A_O'$, and $\mathcal B$ a CPTP map that takes $B_I,B_I'$ to $B_O,B_O'$. A process is then defined as the most general linear transformation that acts trivially on $A_I',B_I',A_O',B_O'$ and takes $\mathcal A$ and $\mathcal B$ to a CPTP map $\mathcal G_{\mathcal A,\mathcal B}$ from $A_I',B_I',P$ to $A_O',B_O',F$, as represented in Fig.~\ref{fig:w_higher-order}. Note that the Hilbert space $F$ cannot signal to any other Hilbert space, and can therefore be interpreted as a global future. Conversely, none of the Hilbert spaces can signal to $P$, which can then be interpreted as a global past.

\begin{figure}[htb]
 \centering
          \begin{tikzpicture}[scale=1.6,baseline=(W.center)]
		\node[ minimum width=0.7cm,minimum height=0.85cm] (W) at (0,0) {$\mathcal G_{\mathcal A,\mathcal B}$};
                \node[] (AI) at (-0.5,0.95) {};
                \node[] (BI) at (0.5,0.95) {};
                \node[] (F) at (0,0.95) {};
                \node[] (AO) at (-0.5,-0.95) {};
                \node[] (BO) at (0.5,-0.95) {};
                \node[] (P) at (0,-0.95) {};
                 \draw[ thick] (-0.8,-0.5) -- (0.8,-0.5) -- (0.8,0.5) -- (-0.8,0.5) -- cycle ;
                \draw[thick,  <-] (AI) -- (-0.5,0.5);
                \draw[thick,  <-] (BI) -- (0.5,0.5);
                \draw[thick,  <-] (F) -- (0,0.5);
                \draw[thick,  ->] (AO) -- (-0.5,-0.5);
                \draw[thick,  ->] (BO) -- (0.5,-0.5);
                \draw[thick,  ->] (P) -- (0,-0.5);
		\node[] (SO) at ([shift={(0.15cm,0.2cm)}]AO) {\footnotesize $A_I'$};
		\node[] (SI) at ([shift={(0.15cm,-0.27cm)}]AI) {\footnotesize $A_O'$};
		\node[] (bobO) at ([shift={(0.15cm,0.2cm)}]BO) {\footnotesize $B_I'$};
		\node[] (bobI) at ([shift={(0.15cm,-0.27cm)}]BI) {\footnotesize $B_O'$};
                \node[] (CI) at ([shift={(0.15cm,-0.27cm)}]F) {\footnotesize $F$};
                \node[] (CI) at ([shift={(0.15cm,0.2cm)}]P) {\footnotesize $P$};
  \end{tikzpicture}
  \begin{tikzpicture}[scale=1.6,baseline=(equal.center)]
  \node[] (equal) at (0,0) {$\quad = \quad$};
  \end{tikzpicture}
          \begin{tikzpicture}[scale=1.6,baseline=(W.center)]
		\node[] (alice) at (-1.2,0)  {$\mathcal A$};
		\node[red] (W) at (0,0) {$W$};
		\node[] (bob) at (1.2,0) {$\mathcal B$};
                \node[] (F) at (0,1.275) {};
                \node[] (P) at (0,-1.275) {};
                 \draw[red, thick] (-1.2,-1) -- (1.2,-1) -- (1.2,-0.65) -- (0.35,-0.65) -- (0.35,0.65) -- (1.2,0.65) -- (1.2,1) -- (-1.2,1) -- (-1.2, 0.65) -- (-0.35,0.65) -- (-0.35,-0.65) -- (-1.2,-0.65) -- cycle;
                \draw[thick, red, ->] (-0.9,0.375) -- (-0.9,0.65);
                \draw[thick,  ->] (-1.5,0.375) -- (-1.5,0.65);
                \draw[thick,  <-] (-1.5,-0.375) -- (-1.5,-0.65);
		\draw[thick] (-1.8,0.375) -- (-0.6,0.375) -- (-0.6,-0.375) -- (-1.8,-0.375) -- cycle;
                \draw[thick, red, <-] (-0.9,-0.375) -- (-0.9,-0.65);
                \draw[thick, red, ->] (0.9,0.375) -- (0.9,0.65);
                \draw[thick,  ->] (1.5,0.375) -- (1.5,0.65);
                \draw[thick,  <-] (1.5,-0.375) -- (1.5,-0.65);
                \draw[thick] (1.8,0.375) -- (0.6,0.375) -- (0.6,-0.375) -- (1.8,-0.375) -- cycle;               
                \draw[thick, red, <-] (0.9,-0.375) -- (0.9,-0.65);
                \draw[thick, red, <-] (0,1.275) -- (0,1);
                \draw[thick, red, ->] (0,-1.275) -- (0,-1);
		\node[red] (aliceO) at ([shift={(0.5cm,0.5cm)}]alice) {\footnotesize $A_O$};
		\node[] (aliceO) at ([shift={(-0.5cm,0.5cm)}]alice) {\footnotesize $A_O'$};
		\node[red] (aliceI) at ([shift={(0.5cm,-0.5cm)}]alice) {\footnotesize $A_I$};
		\node[] (aliceI) at ([shift={(-0.5cm,-0.5cm)}]alice) {\footnotesize $A_I'$};
		\node[red] (bobO) at ([shift={(-0.5cm,0.5cm)}]bob) {\footnotesize $B_O$};
		\node[] (bobO) at ([shift={(0.5cm,0.5cm)}]bob) {\footnotesize $B_O'$};
		\node[red] (bobI) at ([shift={(-0.5cm,-0.5cm)}]bob) {\footnotesize $B_I$};
		\node[] (bobI) at ([shift={(0.5cm,-0.5cm)}]bob) {\footnotesize $B_I'$};
                \node[red] (CI) at ([shift={(0.15cm,-0.15cm)}]F) {\footnotesize $F$};
                \node[red] (CI) at ([shift={(0.15cm,0.15cm)}]P) {\footnotesize $P$};
  \end{tikzpicture}
  \caption{A process $W$ is a bilinear function from the CPTP maps $\mathcal A,\mathcal B$ to a CPTP map $\mathcal G_{\mathcal A,\mathcal B}$. This map takes states from the global past $P$ and auxiliary spaces $A_I',B_I'$ to the global future $F$ and auxiliary spaces $A_O',B_O'$.}\label{fig:w_higher-order}
\end{figure}
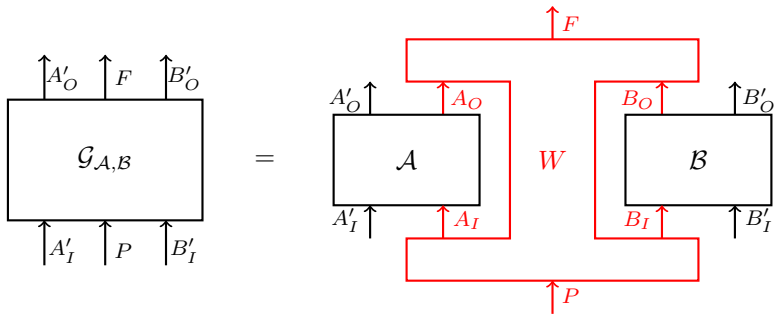

To find a matrix representation for a process, we make use of the Choi-Jamiołkowski (CJ) isomorphism, which we recap in Appendix \ref{sec:choi}. Let then  $A = \choi(\mathcal A)$, $B = \choi(\mathcal B)$, and $G_{A,B} = \choi\de{\mathcal G_{\mathcal A,\mathcal B}}$ be the CJ representations of the CPTP maps $\mathcal A$, $\mathcal B$, and $\mathcal G_{\mathcal A,\mathcal B}$. By linearity we can represent the mapping from $A,B$ to $G_{A,B}$ as 
\be\label{eq:w_higher-order}
G_{A,B} = \tr_{A_IA_OB_IB_O}[ W^{T_{A_IA_OB_IB_O}}(A\otimes B)],
\ee
where $W \in P \otimes A_I \otimes A_O \otimes B_I \otimes B_O \otimes F$ is the process matrix we are defining, $\cdot^{T_{A_IA_OB_IB_O}}$ denotes partial transposition on the subsystems $A_I,A_O,B_I,$ and $B_O$, and identity matrices on the subsystems $P$, $F$, $A_I'$, $A_O'$, $B_I'$, and $B_O'$ are left implicit. This formula can be rewritten in a more convenient way using the link product, which was designed to conveniently express the CJ representation of a composition of quantum operations \cite{chiribella09b}, and we recap in Appendix \ref{sec:link}:
\be
G_{A,B} = W*(A\otimes B).
\ee
We need $G_{A,B}$ to be a valid CPTP map for all valid CPTP maps $A$ and $B$. This imposes the following restrictions on $W$, which we derive in Appendix \ref{sec:valid_w_purification}:
\begin{gather}
 W \ge 0, \\
 \tr W = d_{A_O}d_{B_O}d_P, \\
 W = L_V(W),
\end{gather}
where $L_V$ is a projector on the linear subspace of valid process matrices defined in equation \eqref{eq:projector_lv}.

As an explicit example of a process defined as a higher-order transformation consider a situation where a state in $P = P_1 \otimes P_2 $ is sent to Alice and Bob, and the resulting state from Alice and Bob is sent to $F = F_1 \otimes F_2$. The process matrix is given by $\ketbra{w_\text{state}}$, where
\be
\ket{w_\text{state}} = \Ket{\id}^{P_1A_I}\Ket{\id}^{P_2B_I}\Ket{\id}^{A_OF_1}\Ket{\id}^{B_OF_2},
\ee
and $\Ket{\id}^{P_1A_I}$ is the ``pure'' CJ representation of the identity map between $P_1$ and $A_1$, as defined in Appendix \ref{sec:choi}. Here and throughout the paper we shall use the superscript of a vector to indicate the Hilbert space to where its projector belongs, \eg, $\Ket{\id}^{P_1A_I}$ means that $\KetBra{\id} \in P_1 \otimes A_I$.

A less trivial example is the process where a state in $P$ goes first to Alice, then to Bob, and finally to $F$. Using again the vector representation the process is
\be\label{eq:w_channel}
\ket{w_\text{channel}} = \Ket{\id}^{PA_I}\Ket{\id}^{A_OB_I}\Ket{\id}^{B_OF}.
\ee
Finally, an example of a process that encodes an indefinite causal order is the quantum switch \cite{chiribella09,araujo15witnessing}, which in this representation is:
\be\label{eq:wswitchket}
\ket{w_\text{switch}} = \ket{0}^{P_1}\Ket{\id}^{P_2A_I}\Ket{\id}^{A_OB_I}\Ket{\id}^{B_OF_2}\ket{0}^{F_1} + \ket{1}^{P_1}\Ket{\id}^{P_2B_I}\Ket{\id}^{B_OA_I}\Ket{\id}^{A_OF_2}\ket{1}^{F_1}.
\ee
To connect this version of the quantum switch with the one in Refs.~\cite{chiribella09,araujo15witnessing}, one should prepare the control qubit and send it to $P_1$, and prepare the target qubit and send it to $P_2$. After the process is done the control qubit will be found in $F_1$ (unchanged), and the target qubit in $F_2$.

This definition of processes allows a natural definition of what it means for a process to be ``pure''. Just as we can define unitaries as the most general linear transformations that map pure states to pure states of the same dimension\footnote{More precisely, a linear map $\mathcal E$ is a unitary iff for all pure states $\ket{\psi}$ the transformed state $\mathcal I \otimes \mathcal E (\ketbra{\psi})$ is a pure state of the same dimension.}, we can define pure processes to be the most general linear transformations that map unitaries to unitaries. More precisely:
\begin{definition}\label{def:pure}
 A process $W$ is pure if for all unitaries $\mathcal A$, $\mathcal B$ the resulting transformation  $\mathcal G_{\mathcal A,\mathcal B}$ is also a unitary.
\end{definition}
It turns out that pure processes are unitary transformations from $A_O,B_O,P$ to $A_I,B_I,F$ and can be conveniently represented as vectors, as shown in the following theorem:
\begin{theorem}\label{thm:pure_representation}
 A process $W$ is pure if and only if $W = \KetBra{U_w}$ for some unitary $U_w$.
\end{theorem}
\begin{proof}
 If $W$ is pure, then $\mathcal G_{\mathcal A,\mathcal B}$ must be a unitary in particular when $\mathcal A$ and $\mathcal B$ are SWAPs that map $A_I',A_I$ to $A_O,A_O'$ and $B_I',B_I$ to $B_O,B_O'$. Then
\be
G_{A,B} = W*\de{\KetBra{\id}^{A_I'A_O}\KetBra{\id}^{A_IA_O'} \otimes \KetBra{\id}^{B_I'B_O}\KetBra{\id}^{B_IB_O'}} = W,
\ee
so the resulting transformation $\mathcal G_{\mathcal A,\mathcal B} = \choi^{-1}\de{G_{A,B}} = \choi^{-1}(W) =: \mathcal W$ is just the process itself, with the relabelling $A_I\mapsto A_O'$, $A_O\mapsto A_I'$, $B_I\mapsto B_O'$, and $B_O\mapsto B_I'$, so $\mathcal W$ must be a unitary transformation from $A_O,B_O,P$ to $A_I,B_I,F$. Writing $\mathcal W(\rho) = U_w \rho U_w^\dagger$, we have that its CJ representation is $W = \choi(\mathcal{W}) = \KetBra{U_w}$.

Conversely, if $W = \KetBra{U_w}$, $A = \KetBra{U_a}$, and $B = \KetBra{U_b}$, then
\be
G_{A,B} = \KetBra{U_w}*\de{\KetBra{U_a}\otimes \KetBra{U_b}} = \KetBra{U_g},
\ee
where
\be
U_g^{A_I'B_I'P} = \tr_{A_IB_I}\De{ \de{U_w^{PA_IB_I} \otimes \id^{A_I'B_I'}} \de{\id^{P}\otimes U_a^{A_IA_I'} \otimes U_b^{B_IB_I'} } }.
\ee
Since by assumption $\mathcal G_{\mathcal A,\mathcal B}$ is trace preserving, we have that $\tr_{A_O'B_O'F} G_{A,B} =\id^{A_I'B_I'P}$. Substituting $G_{A,B} = \KetBra{U_g}$ we get
\be
\tr_{A_O'B_O'F} \KetBra{U_g} = U_g^\dagger U_g = \id^{A_I'B_I'P}.
\ee
Since, furthermore, $U_g$ is a square finite-dimensional matrix, it has a right inverse which is equal to its left inverse, so $U_g$ is a unitary, and so is $\mathcal G_{\mathcal A,\mathcal B}$.
\end{proof}

A different definition of purity for process matrices was used in the Appendix A of Ref.~\cite{araujo15witnessing}: there they defined pure processes simply as those that can be written as vectors, i.e., that are proportional to rank-one projectors. We are not going to use this definition because rank-one processes do not necessarily induce reversible transformations between quantum states, and therefore fail to capture the essential feature we want from ``purity''.  This is because for some isometries $V_w$ the matrix $\KetBra{V_w}$ is a valid rank-one process, that however maps unitaries $\mathcal A$ and $\mathcal B$ to a non-unitary isometry $\mathcal G_{\mathcal A,\mathcal B}$, which increases the dimension of the Hilbert space and is therefore not reversible. Moreover, this definition is trivial in the sense that it would make every process purifiable, as can be seen from the discussion in Section \ref{sec:ifandonlyif}.

\section{A purification postulate}\label{sec:postulate}

With the definition of a pure process in hand, we can now define what it means to purify a process:
\begin{definition}\label{def:purification}
 A process $W \in P \otimes A_I \otimes A_O \otimes B_I \otimes B_O \otimes F$ is purifiable if one can recover it from a pure process $S \in P \otimes P' \otimes A_I \otimes A_O \otimes B_I \otimes B_O \otimes F \otimes F'$ by inputting the state\footnote{We can fix the state to be $\ket{0}$ instead of allowing an arbitrary state without loss of generality.} $\ket{0}$ in $P'$ and tracing out $F'$, \ie, if 
\be\label{eq:purification}
W = S*(\ketbra{0}^{P'}\otimes \id^{F'}).
\ee
\end{definition}
Note that there are no restrictions on the dimensions of $P, P', F$, and $F'$, so for instance a pure process $W$ is trivially purifiable by setting $S = W$ and $d_{P'} = d_{F'} = 1$. One can also ask whether a process $W$ with trivial $P$ and $F$ is purifiable by setting $d_P = d_F = 1$, as shown in Fig.~\ref{fig:purification}. This case is in fact the focus of this paper, as most processes considered in the literature have trivial $P$ and $F$ and those are the ones we shall test for purifiability.

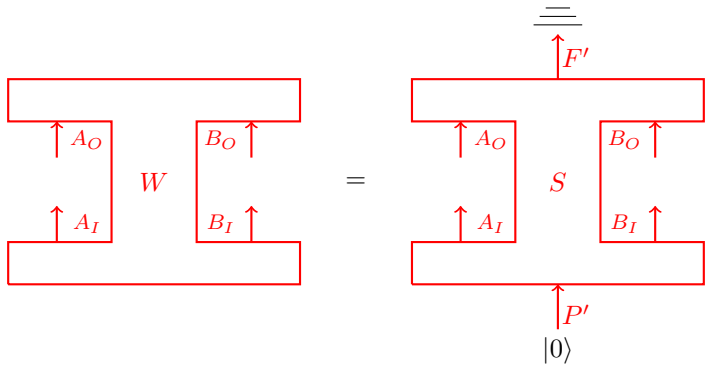
\begin{figure}[htb]
 \centering
          \begin{tikzpicture}[scale=1.6,baseline=(W.center)]
		\node[] (alice) at (-0.8,0) {};
		\node[red, minimum width=0.7cm,minimum height=0.85cm] (W) at (0,0) {$W$};
		\node[] (bob) at (0.8,0) {};
                \node[] (F) at (0,1.3) {};
                \node[] (P) at (0,-1.3) {};
                 \draw[red, thick] (-1.2,-0.85) -- (1.2,-0.85) -- (1.2,-0.5) -- (0.35,-0.5) -- (0.35,0.5) -- (1.2,0.5) -- (1.2,0.85) -- (-1.2,0.85) -- (-1.2, 0.5) -- (-0.35,0.5) -- (-0.35,-0.5) -- (-1.2,-0.5) -- (-1.2,-0.85);
                \draw[thick, red, ->] ([shift={(0.0cm,0.2cm)}]alice.center) -- (-0.8,0.5);
                \draw[thick, red, <-] ([shift={(0.0cm,-0.2cm)}]alice.center) -- (-0.8,-0.5);
                \draw[thick, red, ->] ([shift={(0.0cm,0.2cm)}]bob.center) -- (0.8,0.5);
                \draw[thick, red, <-] ([shift={(0.0cm,-0.2cm)}]bob.center) -- (0.8,-0.5);
		\node[red] (aliceO) at ([shift={(0.25,0.35)}]alice) {\footnotesize $A_O$};
		\node[red] (aliceI) at ([shift={(0.25cm,-0.35cm)}]alice) {\footnotesize $A_I$};
		\node[red] (bobO) at ([shift={(-0.25cm,0.35cm)}]bob) {\footnotesize $B_O$};
		\node[red] (bobI) at ([shift={(-0.25cm,-0.35cm)}]bob) {\footnotesize $B_I$};
  \end{tikzpicture}
  \begin{tikzpicture}[scale=1.6,baseline=(equal.center)]
  \node[] (equal) at (0,0) {$\quad = \quad$};
  \end{tikzpicture}
          \begin{tikzpicture}[scale=1.6,baseline=(W.center)]
		\node[] (alice) at (-0.8,0) {};
		\node[red, minimum width=0.7cm,minimum height=0.85cm] (W) at (0,0) {$S$};
		\node[] (bob) at (0.8,0) {};
                \node[] (F) at (0,1.3) {};
                \node[] (P) at (0,-1.3) {};
                \draw[black] (-0.2,1.3) -- (0.2,1.3);
                \draw[black] (-0.15,1.37) -- (0.15,1.37);
                \draw[black] (-0.1,1.44) -- (0.1,1.44);
                 \draw[red, thick] (-1.2,-0.85) -- (1.2,-0.85) -- (1.2,-0.5) -- (0.35,-0.5) -- (0.35,0.5) -- (1.2,0.5) -- (1.2,0.85) -- (-1.2,0.85) -- (-1.2, 0.5) -- (-0.35,0.5) -- (-0.35,-0.5) -- (-1.2,-0.5) -- (-1.2,-0.85);
                \draw[thick, red, ->] ([shift={(0.0cm,0.2cm)}]alice.center) -- (-0.8,0.5);
                \draw[thick, red, <-] ([shift={(0.0cm,-0.2cm)}]alice.center) -- (-0.8,-0.5);
                \draw[thick, red, ->] ([shift={(0.0cm,0.2cm)}]bob.center) -- (0.8,0.5);
                \draw[thick, red, <-] ([shift={(0.0cm,-0.2cm)}]bob.center) -- (0.8,-0.5);
                \draw[thick, red, <-] (F) -- (0,0.85);
                \draw[thick, red, ->] (P) -- (0,-0.85);
		\node[red] (SO) at ([shift={(0.25,0.35)}]alice) {\footnotesize $A_O$};
		\node[red] (SI) at ([shift={(0.25cm,-0.35cm)}]alice) {\footnotesize $A_I$};
		\node[red] (bobO) at ([shift={(-0.25cm,0.35cm)}]bob) {\footnotesize $B_O$};
		\node[red] (bobI) at ([shift={(-0.25cm,-0.35cm)}]bob) {\footnotesize $B_I$};
                \node[red] (CI) at ([shift={(0.15cm,-0.27cm)}]F) {$F'$};
                \node[red] (CI) at ([shift={(0.15cm,0.2cm)}]P) {$P'$};
                \node[] (CI) at ([shift={(0.0cm,-0.1cm)}]P) {$\ket{0}$};
  \end{tikzpicture}
  \caption{A process $W$ with trivial $P$ and $F$ is purifiable if it can be recovered from a pure process $S$ by inputting the state $\ket{0}$ in $P'$ and tracing out $F'$.}\label{fig:purification}
\end{figure}

We propose then the purification postulate: \textit{a process is physical only if it is purifiable}. This postulate is motivated by the fact that a non-purifiable process would fundamentally map unitaries onto isometries or non-unitary CPTP maps, destroying the reversibility of the theory. Reversibility, in its turn, is a cherished principle \cite{chiribella10}: all fundamental physical theories are time reversible\footnote{Note that to treat the collapse of the wave function during a measurement as an irreversible process one needs to use objective collapse models \cite{ghirardi86,ghirardi90,bassi13} instead of standard quantum mechanics. We are taking the view that collapse is not a physical process.} and the hint that the formation and evaporation of a black hole might be an irreversible evolution is one of the major problems in modern physics \cite{hawking76,harlow14}. Moreover, in all reconstructions of quantum mechanics to date \cite{hardy01,dakic09,masanes11,chiribella11,barnum14,hoehn14b,hoehn15} reversibility has been used as a central ingredient. This supports the idea that it is indeed a fundamental part of quantum mechanics, and it should not be done away with lightly. Furthermore, although there is speculation that in a full theory of quantum gravity new degrees of freedom are created by the expansion of the Universe \cite{jacobson99,bojowald07,gielen16}, in concrete models these degrees of freedom are born in the vacuum state \cite{hoehn14}, making the time evolution actually reversible. In any case, current inflationary models using quantum field theory on a curved expanding spacetime do not feature creation of new degrees of freedom \cite{mukhanov05}.

As a sanity check, note that all processes known to be physical -- all processes where the order between the parties is fixed, \eg $\ket{w_\text{channel}}$ (equation \eqref{eq:w_channel}), controlled by a quantum system, \eg $\ket{w_\text{switch}}$ (equation \eqref{eq:wswitchket}), or incoherently mixed -- are purifiable.

\section{Necessary and sufficient conditions for purification}\label{sec:ifandonlyif}

We shall now derive necessary and sufficient conditions for a process $W$ with trivial $P$ and $F$ to be purifiable, \ie, to be  recoverable via equation \eqref{eq:purification} from a pure process $S$. From Theorem \ref{thm:pure_representation}, we know that $S$ is pure iff $S = \KetBra{U_s}^{P'A_IA_OB_IB_OF'}$ for some unitary $U_s$. Defining 
\be
\ket{w_\psi}^{A_IA_OB_IB_OF'} := \bra{\psi^*}^{P'} \Ket{U_s}^{P'A_IA_OB_IB_OF'}
\ee
and noting that
\be
\KetBra{U_s}^{P'A_IA_OB_IB_OF'}*(\ketbra{\psi}^{P'}\otimes \id^{F'}) = \ketbra{w_\psi}^{A_IA_OB_IB_OF'} * \id^{F'} = \tr_{F'} \ketbra{w_\psi}^{A_IA_OB_IB_OF'},
\ee
we can rewrite equation \eqref{eq:purification} as
\begin{equation}\label{eq:stinespring_pure}
W^{A_IA_OB_IB_O} = \tr_{F'} \ketbra{w_0}^{A_IA_OB_IB_OF'},
\end{equation}
where $\ket{w_0}^{A_IA_OB_IB_OF'} = \bra{0}^{P'} \Ket{U_s}^{P'A_IA_OB_IB_OF'}$. 

We can also state the condition that $\Ket{U_s}$ is a valid process purely in terms of the vectors\footnote{From now on we shall write $\Ket{U_s}^{P'A_IA_OB_IB_OF'}$ and $\ket{w_\psi}^{A_IA_OB_IB_OF'}$ without their superscripts for brevity.} $\ket{w_\psi}$. To do this, first note that, as discussed in Appendix \ref{sec:valid_w_purification}, we only need to consider the ancillary spaces $A_I',A_O',B_I'$, and $B_O'$ in the definition of processes to conclude that a process matrix must be positive semidefinite. Since $\KetBra{U_s}$ must already be positive semidefinite from its form, the definition of process reduces to saying that $\KetBra{U_s}$ is valid iff $\KetBra{U_s}*(A\otimes B)$ is a CPTP map for all CPTP maps $A \in A_I \otimes A_O$ and $B \in B_I \otimes B_O$. 

In its turn, $\KetBra{U_s}*(A\otimes B)$ is a CPTP map iff for all input states $\ket{\psi}^{P'}$ its output $\Big[\KetBra{U_s}*(A\otimes B)\Big]*\ketbra{\psi}^{P'}$ is a valid quantum state. Note, however, that 
\be \Big[\KetBra{U_s}*(A\otimes B)\Big]*\ketbra{\psi}^{P'} = \Big(\KetBra{U_s}*\ketbra{\psi}^{P'}\Big)*(A\otimes B),\ee
and the condition that the right hand side is a valid quantum state for all $\ket{\psi}$, $A$, and $B$ is precisely the condition that the process with trivial $P'$ given by $(\KetBra{U_s}*\ketbra{\psi}^{P'}) = \ketbra{w_\psi}$ produces a valid CPTP map (with trivial $P'$) when linked with the CPTP maps $A$ and $B$, which means that $\ketbra{w_\psi}$ is a valid process. Writing out the validity conditions explicitly, we have
\begin{equation}\label{eq:validity_pure}
\forall\ket{\psi}\quad \ketbra{w_\psi} = L_V(\ketbra{w_\psi})\mathand\tr \ketbra{w_\psi} = d_{A_O}d_{B_O},
\end{equation}
where $L_V$ is the projector onto the valid subspace defined in equation \eqref{eq:projector_lv} particularized for $d_{P'} = 1$.

With this, we reduce the problem of purifying a process $W$ to that of finding a set of vectors $\ket{w_\psi}$ such that equations \eqref{eq:stinespring_pure} and \eqref{eq:validity_pure} are satisfied. We can simplify it further by noting that equation \eqref{eq:stinespring_pure} is just the purification of a positive matrix, which admits a simple solution. Let then the eigendecomposition of $W$ be 
\be
W = \sum_{i=0}^{r-1} \lambda_i \ketbra{\lambda_i}^{A_IA_OB_IB_O},
\ee
where $r$ is the rank of $W$, and $\lambda_i, \ket{\lambda_i}$ its nonzero eigenvalues and corresponding eigenvectors. Then a valid purification for it is
\be\label{eq:w0}
\ket{w_0} = \sum_{i=0}^{r-1} \sqrt{\lambda_i}\ket{\lambda_i}^{A_IA_OB_IB_O}\ket{i}^{F'},
\ee
which is unique modulo some isometry $V$ on the purifying system. But this isometry has no effect on the validity of $\ketbra{w_\psi}$, as it only affects the state output by the process, and an isometry maps valid quantum states to valid quantum states. This implies that we can choose without loss of generality the isometry to be the identity, and the dimension of the purifying system $d_{F'} = d_{P'}$ to be the rank of $W$. This allows us to use an $r^2$-dimensional basis for $P'$ to simplify condition \eqref{eq:validity_pure}. Choosing $\{\ket{i}\bra{j}\}_{i,j=0}^{r-1}$ as this basis, and rewriting the equations explicitly in terms of the vectors $\ket{w_i} = \bra{i}^{P'} \Ket{U_s}^{P'A_IA_OB_IB_OF'}$, the condition translates to
\be
\forall i,j \quad L_V^\perp(\ket{w_i}\bra{w_j}) = 0 \mathand \braket{w_i}{w_j} = d_{A_O}d_{B_O} \delta_{ij}, 
\ee
where for brevity we are using the projector $L_V^\perp := \id - L_V$.

We summarize the results of this Section in the following theorem:
\begin{theorem}\label{thm:purification}
 A process matrix $W$ of rank $r$ and eigendecomposition
 \be
 W = \sum_{i=0}^{r-1} \lambda_i \ketbra{\lambda_i}^{A_IA_OB_IB_O}
 \ee
 is purifiable if and only if there exists a set of vectors $\{\ket{w_i}\}_{i=0}^{r-1}$ such that
 \be
 \ket{w_0} = \sum_{i=0}^{r-1} \sqrt{\lambda_i}\ket{\lambda_i}^{A_IA_OB_IB_O}\ket{i}^{F'}
 \ee
 and
 \be \forall i,j \quad L_V^\perp(\ket{w_i}\bra{w_j}) = 0 \mathand \braket{w_i}{w_j} = d_{A_O}d_{B_O} \delta_{ij}.
 \ee
 If they exist, a process $S = \KetBra{U_s}$ that purifies $W$ is given by
 \be
 \forall i \quad \bra{i}^{P'}\Ket{U_s}^{P'A_IA_OB_IB_OF'} := \ket{w_i}
 \ee
 \end{theorem}

\section{Necessary condition for purification}\label{sec:onlyif}

It is not easy to find a set of vectors $\{\ket{w_i}\}$ that satisfies the conditions of Theorem \ref{thm:purification}, as the constraint $L_V^\perp(\ket{w_i}\bra{w_j}) = 0$ is quadratic on them. We shall not solve this problem in full, but rather prove an upper bound on the number of vectors that can satisfy those conditions for a given process. If this upper bound is smaller than the rank of the process, this is a proof that this process is not purifiable.

We shall start by characterising the vector space $V_W$ formed by the vectors $\ket{v}$ such that $L_V^\perp(\ket{v}\bra{w_0}) = 0$. Since this condition is linear on $\ket{v}$ it is straightforward to do it. Furthermore, since this is one of the conditions in Theorem \ref{thm:purification}, it should be clear that the set of vectors $\{\ket{w_i}\}$, if it exists, must belong to $V_W$. Restricting our attention to this (hopefully small) subspace makes it easier to consider the other, non-linear, conditions.

To construct an orthonormal basis for $V_W$, first we transform the projector $L_V^\perp$ to act on the ``double-kets'' of process matrices, \ie, we define the matrix $\Pi_{L_V^\perp}$ such that\footnote{$\Pi_{L_V^\perp}$ can be explicitly written as $\Pi_{L_V^\perp} = \choi(L_V^\perp)^R$, where $R$ is the reshuffling operation \cite{zyczkowski06}.}
\be
\forall M \in A_I\otimes A_O\otimes B_I\otimes B_O\otimes F' \quad \Pi_{L_V^\perp} \Ket{M} = \bigKet{L_V^\perp(M)}.
\ee
We have then 
\be
L_V^\perp(\ket{v}\bra{w_0}) = 0 \iff \Pi_{L_V^\perp} \ket{w_0^*}\ket{v} = 0,
\ee
so if we define
\be
O_W := \big(\bra{w_0^*} \otimes \id\big) \Pi_{L_V^\perp} \big(\ket{w_0^*} \otimes \id\big)
\ee
we have that
\be
O_W \ket{v} = 0 \iff \Pi_{L_V^\perp} \ket{w_0^*}\ket{v} = 0,
\ee
so the null eigenvectors of $O_W$ span $V_W$, and we can use them to restrict the projector $\Pi_{L_V^\perp}$ to the subspace $V_W^* \otimes V_W$. 

To do that, let $\{\ket{r_i}\}_{i=0}^{\dim(V_W)-1}$ be the (orthonormal) null eigenvectors of $O_W$, and let $\{\ket{i}\}_{i=0}^{\dim(V_W)-1}$ be a generic computational basis of dimension $\dim(V_W)$. The restriction is then done via the isometry $R = \sum_{i=0}^{\dim(V_W)-1} \ket{i}\bra{r_i}$, which maps a vector space of dimension $d_{A_I}d_{A_O}d_{B_I}d_{B_O}\rank(W)$ to a vector space of dimension $\dim(V_W)$, usually drastically smaller. The restricted operator $\Pi_{L_V^\perp|V_W}$ is then given by
\be
\Pi_{L_V^\perp|V_W} = (R^* \otimes R)\ \Pi_{L_V^\perp}\, (R^T \otimes R^\dagger).
\ee
Using $\Pi_{L_V^\perp|V_W}$ we can now particularize the condition that $L_V^\perp(\ket{a}\bra{b}) = 0$ to vectors $\ket{a}$ and $\ket{b}$ inside $V_W$. Let $\ket{t} = R \ket{a}$ and $\ket{u} = R\ket{b}$. Then $L_V^\perp(\ket{a}\bra{b}) = 0$ iff $\Pi_{L_V^\perp|V_W}\ket{u^*}\ket{t} = 0$. Let then $\{\ket{m_k}\}_{k=0}^{d_m-1}$ be the set of eigenvectors of $\Pi_{L_V^\perp|V_W}$ with nonzero eigenvalue\footnote{It is possible that this set is empty. In this case all vectors in $V_W$ respect the conditions of Theorem \ref{thm:purification}, and the process is purifiable iff $\dim(V_W) \ge \rank(W)$.}. Then
\be
\Pi_{L_V^\perp|V_W}\ket{u^*}\ket{t} = 0 \iff\forall k \quad \bra{u^*}\bra{t}\ket{m_k} = 0.
\ee
We want to rewrite the inner product $\bra{u^*}\bra{t}\ket{m_k}$ in a more convenient way. For that, note that $\bra{\psi}A\ket{\phi} = \bra{\phi^*}\bra{\psi}\Ket{A^T}$, so $\bra{u^*}\bra{t}\ket{m_k} = \bra{t}M_k\ket{u}$ for matrices $M_k$ such that $\Ket{M_k^T} = \ket{m_k}$. These matrices are in general not Hermitian, so for convenience we define $C_k = M_k + M_k^\dagger$ and $C_{k+d_m} = i(M_k-M_k^\dagger)$. Then
\be
\Pi_{L_V^\perp|V_W}\ket{u^*}\ket{t} = 0 \iff\forall k \quad \bra{t}C_k\ket{u}=0.
\ee
Let now $d_{C_k}$ be the dimension of the largest subspace such that for all vectors $\ket{t},\ket{u}$ in this subspace $\bra{t}C_k\ket{u}=0$. We can easily calculate $d_{C_k}$ using the the null-square lemma proved in Appendix \ref{sec:nullsquare}: let $n_{k+},n_{k-}$, and $n_{k0}$ be the number of positive, negative, and null eigenvalues of $C_k$. Then 
\be
d_{C_k} = n_{k0} + \min \{n_{k+},n_{k-}\},
\ee
and a simple upper bound on the dimension of the largest subspace of vectors that respect the conditions of Theorem \ref{thm:purification} is
\begin{equation}
 d_\text{max}(W) := \min_i d_{C_i}.
\end{equation}
Therefore if $d_\text{max}(W) < \rank(W)$ the process $W$ is not purifiable.

\section{Examples}\label{sec:examples}

We shall now apply the necessary condition derived in the previous Section to several process matrices from the literature. The calculation of $d_\text{max}$ for all the matrices in this Section was done numerically with the MATLAB code available as the ancillary file \texttt{purification.m} on the arXiv.

In this Section we shall omit superscripts that identify subsystems and tensor products to avoid clutter. The expression (\eg) $Z \id X Z$ should be understood as $Z^{A_I} \otimes \id^{A_O} \otimes X^{B_I} \otimes Z^{B_O}$.

Let \begin{equation} \label{eq:w_ocb}
  W_\text{OCB} = \frac{1}{4} \left[ \id\id\id\id + \frac{\id Z Z \id + Z \id X Z }{\sqrt{2}} \right]
\end{equation}
be the process introduced in Ref.~\cite{oreshkov12}. It was proven, under a restriction on the allowed operations, to produce the maximal violation of the original causal inequality for any dimension \cite{brukner14}. Since $\rank(W_{\text{OCB}}) = 8$ and $d_\text{max}(W_{\text{OCB}}) = 5$ it is not purifiable.

Let \begin{multline}
W_\text{max} = \frac14 \Big[ \id\id\id\id +  a_0 \, Z\id Z\id  - a_1\big(Z\id\id\id + \id\id Z\id\big) 
- a_2\big(Z\id\id Z + \id Z Z\id\big) + a_3\big(Z\id ZZ + ZZZ\id\big) \\
+ a_4\big(Z\id XX - Z\id YY  + XXZ\id - YYZ\id\big)  \Big],
\end{multline}
where $a_0 \approx 0.2744$, $a_1 \approx 0.2178$, $a_2 \approx 0.3628$, $a_3 \approx 0.3114$, and $a_4 \approx 0.2097$. This process was introduced in Appendix C of Ref.~\cite{branciard15simplest} and conjectured to produce the maximal violation of GYNI and LGYNI inequalities for its dimension. Since $\rank(W_{\text{max}}) = 11$ and $d_\text{max}(W_{\text{max}}) = 10$ it is not purifiable.

Let \begin{equation} \label{eq:w_opt}
  W_\text{opt} = \frac{1}{4} \left[ \id\id\id\id + \frac1{\sqrt3} \, Z\id X Z +  \frac{\sqrt3-1}3 \big(\id X X \id + \id Y Y \id + \id Z Z \id\big) \right]
\end{equation}
be the process introduced in Ref.~\cite{feix16}. By itself it cannot violate any causal inequalities, but it becomes able to do so when it is extended with an entangled ancilla. However, when admixed with a small amount
of noise, the known violations disappear, so this noisy version was conjectured to be unable to violate any causal inequalities \cite{feix16}, that is, to be ``extensibly causal'' \cite{oreshkov15}. Since extensibly causal processes seem physically reasonable, but we have no physical interpretation for $W_\text{opt}$ or its noisy version, it would very interesting to test its purifiability to gather evidence about its physicality. Unfortunately, since $\rank(W_{\text{opt}}) = 12$ and $d_\text{max}(W_{\text{opt}}) = 17$ our test is inconclusive. This is also the case for the noisy version, which has rank $16$ and $d_\text{max}$ equal to $41$.

\section{Counter-example}\label{sec:counterexample}

Since the processes examined in the previous Section either are not purifiable or have unknown status, it might be tempting to conjecture that all processes that can violate causal inequalities are not purifiable. While this might be true for the bipartite case, there exists a counter-example for the general case: a tripartite process which can violate causal inequalities and is purifiable. Note that although Definition \ref{def:pure}, Theorem \ref{thm:pure_representation}, and Definition \ref{def:purification} were stated explicitly only for the bipartite case, their multipartite generalisation is straightforward, and we shall use it in this Section.

This process, due to Ref.~\cite{araujo14private}, and first published in Ref.~\cite{baumeler16}, can be concisely described as a function from $A_O,B_O,C_O$ to $A_I,B_I,C_I$ that incoherently maps the basis states $\ket{a,b,c}$ to the basis states $\ket{\lnot b \land c,\lnot c \land a,\lnot a \land b}$, where $\lnot$ and $\land$ represent logical negation and logical conjunction. It can be represented in a more cumbersome way as the process matrix
\be
W_\text{det} = \sum_{abc} \ketbra{a,b,c} \otimes \ketbra{\lnot b \land c,\lnot c \land a,\lnot a \land b},
\ee
where unlike in the other processes, we wrote the subsystems in the order $A_O$,$B_O$,$C_O$,$A_I$,$B_I$,$C_I$, to match the function described above. The purification of this process, due to Ref.~\cite{baumeler15private}, can be done with the standard trick to turn an irreversible function into a reversible one, that is, changing the function from $\ket{x} \mapsto \ket{f(x)}$ to $\ket{x}\ket{y} \mapsto \ket{x}\ket{y\oplus f(x)}$. In terms of our function this means mapping $\ket{a,b,c}\ket{i,j,k}$ to $\ket{a,b,c}\ket{i\oplus\lnot b \land c, j\oplus\lnot c \land a, k\oplus \lnot a \land b}$. This reversible function is then the purification of $W_\text{det}$, and can be written as the process vector
\be\label{eq:wdetket}
\ket{w_\text{det}} = \sum_{\substack{abc\\ijk}}\ket{a,b,c}\ket{i,j,k} \otimes \ket{a,b,c}\ket{i\oplus\lnot b \land c, j\oplus\lnot c \land a, k\oplus \lnot a \land b},
\ee
where the subsystems are written in the order $A_O$,$B_O$,$C_O$,$P$,$F$,$A_I$,$B_I$,$C_I$. Note that unlike the other subsystems, $P$ and $F$ consist of three qubits each.

\section{Conclusion}

We have proposed purifiability as a necessary condition for the physicality of a process, motivated by considerations that are independent of the process matrix formalism. This allows us for the first time to declare a large class of processes to be unphysical. This can be made part of an axiomatic, top-down approach to determining the set of physical processes. Ideally, one should marry it with a bottom-up approach, that shows processes to be physical by actually constructing laboratory implementations for them. As of yet, however, there are still processes which are neither excluded by the purification principle nor known to be implementable, showing that more work must be done in order to consummate this marriage.

In this grey zone we find, for example, $\ket{w_\text{det}}$ (equation \eqref{eq:wdetket}), showing that the set of purifiable processes is rather rich, and that the purification principle does not rule out the violation of causal inequalities in Nature. It might still be the case, however, that processes that cannot violate causal inequalities, \ie, extensibly causal processes, are always purifiable. In order to decide this, it would be interesting to know whether the noisy version of $W_\text{opt}$ (equation \eqref{eq:w_opt}) is purifiable.

\section{Acknowledgements}

We thank Ämin Baumeler, Fabio Costa, Paolo Perinotti, and Stefan Wolf for useful discussions. We acknowledge support from the Austrian Science Fund (FWF) through the Special Research Programme FoQuS, the Doctoral Programme CoQuS and the projects No. P-24621 and I-2526. This work has been supported by the Excellence Initiative of the German Federal and State Governments (Grant ZUK 81). This publication was made possible through the support of a grant from the John Templeton Foundation. The opinions expressed in this publication are those of the authors and do not necessarily reflect the views of the John Templeton Foundation. 

\printbibliography

\appendix

\section{The Choi-Jamiołkowski isomorphism}\label{sec:choi}

We use two levels of the Choi-Jamiołkowski isomorphism. The first level is the ``pure'' CJ isomorphism, used to represent matrices as vectors. The CJ operator of a matrix $A:\mathcal H_I\to\mathcal H_O$ is its ``double-ket'' \cite{royer_wigner_1991,braunstein_universal_2000}:
\be
\Ket{A} := \id \otimes A \Ket{\id} = \sum_{i=0}^{d_{\mathcal H_I}-1} \ket{i}A\ket{i}.
\ee
The second level is the ``mixed'' CJ isomorphism, used to represent linear operators that act on matrices as matrices themselves. The CJ operator of a map $\mathcal M : \mathcal L(\mathcal H_I) \to \mathcal L(\mathcal H_O)$ is
\be 
\choi(\mathcal M) := \mathcal I \otimes \mathcal M (\KetBra{\id}) = \sum_{i,j=0} ^{d_{\mathcal H_I}-1}\ket{i}\bra{j} \otimes \mathcal M (\ket{i}\bra{j}).
\ee
Note that this convention differs from the one in \cite{oreshkov12,araujo15witnessing} by a transposition. This is necessary to avoid representing completely positive maps with a non-positive matrix, as would happen with the convention from \cite{oreshkov12,araujo15witnessing} when, for example, we take in equation \eqref{eq:w_higher-order} the process as being identity from $P$ to $A_I$ and identity from $A_O$ to $F$, and the map $\mathcal A$ to be a SWAP.

\section{The link product}\label{sec:link}

The motivation for defining the link product is to express the CJ operator of a composition of two linear maps $\mathcal M$ and $\mathcal N$ directly as a function of the CJ operators of the individual maps \cite{chiribella09b}, as 
\be
\choi(\mathcal N\circ \mathcal M) = \choi(\mathcal N) * \choi(\mathcal M).
\ee
We do not need, however, to explicitly use the CJ isomorphism, as the link product can be more simply defined as an abstract operation on two matrices. Namely, the link product between two operators $M \in \bigotimes_{i\in I} A^i$ and $N \in \bigotimes_{i\in J} A^i$ is
\be
N * M := \tr_{I\cap J} [(\id^{J\setminus I} \otimes M^{T_{I\cap J}})(N \otimes \id^{I\setminus J})],
\ee
where $\cdot^{T_{I\cap J}}$ denotes partial transposition on the subsystems $I\cap J$.

Note that when $I\cap J = \emptyset$ the link product reduces to $N*M = N\otimes M$ and, conversely, when $I\cap J = I \cup J$, it reduces to $N*M = \tr M^T N$.

\section{Conditions for validity}\label{sec:valid_w_purification}

In this Appendix we derive the constraints on $W$ such that the map 
\be
G_{A,B} =  W*(A\otimes B)
\ee
is a valid CPTP map for all valid CPTP maps $A$ and $B$. This derivation follows closely the one presented in Appendix B of Ref.~\cite{araujo15witnessing}. Recall that a linear map $\mathcal A : A_I \to A_O$ is completely positive and trace preserving if and only if its CJ operator $A = \choi(\mathcal A)$ is positive and $\tr_{A_O}A = \id^{A_I}$. Therefore we need that
\begin{gather} 
\tr_{A_O' B_O' F} W*(A\otimes B) = \id^{A_I' B_I' P} \mathand W*(A\otimes B) \ge 0 \\
\forall A,B \quad \text{s.t.} \quad A\ge 0,\quad\tr_{A_OA_O'}A=\id^{A_IA_I'} \\
B\ge 0,\quad\tr_{B_OB_O'}B=\id^{B_IB_I'}.
\end{gather}
First we show that $W*(A\otimes B) \ge 0$ iff $W \ge 0$. To see that, note that when $\mathcal A$ and $\mathcal B$ are swaps we have $W*(A\otimes B) = W$, so $W \ge 0$ is a necessary condition for $G_{A,B} \ge 0$. To see that it is also sufficient, just note that the link product of two positive operators is always positive.

To get the constraints on $W$ such that $\tr_{A_O' B_O' F} W*(A\otimes B) = \id^{A_I' B_I' P}$, first note that the positivity of $A,B$, and $W$ is not relevant for that, since the set of positive operators is a full-dimensional subset of the space of Hermitian operators. Therefore we only need the normalisation condition
\begin{gather} 
\tr_{A_O' B_O' F} W*(A\otimes B) = \id^{A_I' B_I' P} \\
\forall A,B \quad \text{s.t.} \quad \tr_{A_OA_O'}A=\id^{A_IA_I'},\quad\tr_{B_OB_O'}B=\id^{B_IB_I'}.
\end{gather}
Note also that the set of operators of the form $\tr_{A_I'A_O'}A$ for CPTP $A$ does not depend\footnote{It would if the CPTP maps $A$ were restricted to be unitaries.} on the dimension of $A_I'$ and $A_O'$, so we can without loss of generality set $d_{A_I'} = d_{A_O'} = d_{B_I'} = d_{B_O'} = 1$, simplifying this condition to
\begin{gather} 
\tr_{F} W*(A\otimes B) = \id^{P} \\
\forall A,B \quad \text{s.t.} \quad \tr_{A_O}A=\id^{A_I},\quad\tr_{B_O}B=\id^{B_I}.
\end{gather}
Note also that
\be
\tr_{A_O}A = \id^{A_I} \iff A = a - {}_{A_O}a + \frac{\id^{A_IA_O}}{d_{A_O}}
\ee
for some Hermitian operator $a$, where we are using the shorthand notation
\begin{equation}
 _X W = \frac{\id^{X}}{d_X} \otimes \tr_X W.
\end{equation}
Using this in Alice and Bob's sides allows us to reduce the normalization condition to
\be\label{eq:normalization_final}
\forall a,b\quad (\tr_F W)*\De{\de{a - {}_{A_O}a + \frac{\id^{A_IA_O}}{d_{A_O}}}\otimes \de{b - {}_{B_O}b + \frac{\id^{B_IB_O}}{d_{B_O}}}}= \id^{P}.
\ee
For $a=b=0$ this yields the condition
\be\label{eq:trace_condition}
(\tr_F W) * \de{\frac{\id^{A_IA_O}}{d_{A_O}}\otimes \frac{\id^{B_IB_O}}{d_{B_O}}} =  \id^{P}.
\ee
For $b=0$ and $a=0$, respectively, equation \eqref{eq:trace_condition} together with \eqref{eq:normalization_final} imply that
\begin{gather}
 \forall a\quad (\tr_F W)*\De{(a - {}_{A_O}a )\otimes \frac{\id^{B_IB_O}}{d_{B_O}}}= 0 \label{eq:local_loop_a} \\ 
 \forall b\quad (\tr_F W)*\De{\frac{\id^{A_IA_O}}{d_{A_O}}\otimes (b - {}_{B_O}b)}= 0 \label{eq:local_loop_b}
\end{gather}
which then imply that
\be\label{eq:global_loop}
\forall a,b\quad (\tr_F W)*[(a - {}_{A_O}a)\otimes (b - {}_{B_O}b)]= 0.
\ee
Using the fact that the maps ${}_{A_O}$ and ${}_{B_O}$ are self-adjoint under the Hilbert-Schmidt inner product, we see that conditions \eqref{eq:local_loop_a}-\eqref{eq:global_loop} are equivalent to
\begin{gather}
 {}_{B_IB_OF}W = {}_{A_OB_IB_OF}W \label{eq:local_a}, \\
 {}_{A_IA_OF}W = {}_{A_IA_OB_OF}W \label{eq:local_b}, \\
 {}_{F}W = {}_{A_OF}W + {}_{B_OF}W - {}_{A_OB_OF}W. \label{eq:global}
 \end{gather}
These conditions, together with \eqref{eq:trace_condition}, completely characterize the affine subspace to which valid process matrices belong. It is, however, convenient to rewrite them so that we characterize this affine subspace as a linear subspace together with a single constraint on the trace of $W$. To do that, note that conditions \eqref{eq:local_a}-\eqref{eq:global} define linear subspaces, so we only need to rewrite condition \eqref{eq:trace_condition}. It is easy to see that it is equivalent to:
\begin{gather}
 {}_{A_IA_OB_IB_OF} W = {}_{PA_IA_OB_IB_OF}W \label{eq:post-selection_P},\\
 \tr W = d_Pd_{A_O}d_{B_O}.
\end{gather}
So the linear subspace of valid process matrices is the intersection of the linear subspaces defined in the equations \eqref{eq:local_a}-\eqref{eq:post-selection_P}. Since the projectors onto these subspaces commute, we can obtain the projector onto the subspace of valid process matrices $L_V$ simply by composing them. Elementary manipulations give us then
\begin{multline}\label{eq:projector_lv}
 L_V(W) = W - {}_{F}W + {}_{A_OF}W + {}_{B_OF}W - {}_{A_OB_OF}W - {}_{A_IA_OF}W + {}_{A_IA_OB_OF}W \\- {}_{B_IB_OF}W + {}_{A_OB_IB_OF}W - {}_{A_IA_OB_IB_OF}W + {}_{PA_IA_OB_IB_OF}W.
\end{multline}
To summarize the results of this Appendix, $W$ is a valid process matrix iff
\begin{gather}
 W \ge 0, \\
 \tr W = d_Pd_{A_O}d_{B_O}, \\
 W = L_V(W).
\end{gather}

\section{The null-square lemma}\label{sec:nullsquare}

The purpose of this Appendix is to prove that, for any Hermitian matrix $C$, the greatest null square submatrix of $C$ achievable via a change of basis has size $d_C\equiv\min(n_{+},n_{-})+n_0$, where $n_{+},n_{-},n_0$ denote, respectively, the dimensionality of the subspaces spanned by the eigenvectors of $C$ with positive, negative and zero eigenvalues.

We will first establish that the above quantity is an upper bound for the greatest zero-square submatrix. Let then $S$ be a subspace with the property that $\bra{v}C\ket{w}=0$, for $\ket{v},\ket{w}\in S$, and denote by $\Pi_0$ the projector onto the null space of $C$. 

Note that $\mbox{dim}(S)\leq \mbox{dim}(S_0)+\mbox{dim}(S_{+-})$, where $S_0\equiv\Pi_0 S$ and $S_{+-}\equiv(\id-\Pi_0)S$. Notice as well that  $\bra{v}C\ket{w}=0$, for $\ket{v},\ket{w}\in S_{+-}$. Clearly, the dimension of $S_0$ is upper bounded by $n_0$, hence all we need to do is show that the dimension of $S_{+-}$ is upper bounded by $\min(n_{+},n_{-})$. 

Let $\{\ket{u^{(i)}}\}_{i}$ be a basis for $S_{+-}$, and let $\{\lambda_k\}_{k=1}^{n_{+}}$ and $\{-\mu_k\}_{k=1}^{n_{-}}$ be the set of positive and negative eigenvalues of $C$ (where $\mu_k \ge 0$), with corresponding eigenvectors $\{\ket{\psi_k}\}_{k=1}^{n_{+}}$ and $\{\ket{\phi_k}\}_{k=1}^{n_{-}}$. Then, 
\be
\ket{u^{(i)}}=\sum_{k=1}^{n_{+}}v^{(i)}_k\ket{\psi_k}+\sum_{k=1}^{n_{-}}w^{(i)}_k\ket{\phi_k}.
\ee
From the condition $\bra{u^{(i)}}C\ket{u^{(j)}}=0$ we have that
\be
\sum_{k=1}^{n_{+}}(v^{(i)}_k)^*v^{(j)}_k\lambda_k-\sum_{k=1}^{n_{-}}(w^{(i)}_k)^*w^{(j)}_k\mu_k=0,
\ee
or, in other words, $\braket{\tilde{v}^{(i)}}{\tilde{v}^{(j)}}=\braket{\tilde{w}^{(i)}}{\tilde{w}^{(j)}}$, with $\tilde{v}^{(i)}_k\equiv\sqrt{\lambda_k}v^{(i)}_k$, $\tilde{w}^{(i)}_k\equiv\sqrt{\mu_k}w^{(i)}_k$. This implies that there exists an isometry $U$ such that $\ket{\tilde{v}^{(i)}}=U\ket{\tilde{w}^{(i)}}$. The vectors $\ket{\tilde{v}^{(i)}}\oplus\ket{\tilde{w}^{(i)}}$ (and thus $\ket{u^{(i)}}$) are generated by applying a linear transformation on $\ket{\tilde{w}^{(i)}}$ and so the space they span cannot be greater than $n_{-}$. An identical argument shows that the space spanned by $\{\ket{u^{(i)}}\}_{i=1}^m$ cannot be greater than $n_{+}$ either. We have just proven that $d_C$ is an upper bound for the size of the null square.

To prove that it is also a lower bound, note that the $\min(n_{+},n_{-})$ linearly independent vectors $\{\ket{u_i}\equiv \lambda_k^{-1/2}\ket{\psi_k}+\mu_k^{-1/2}\ket{\phi_k}:i=1,...,\min(n_{+},n_{-})\}$ satisfy $\bra{u_i}C\ket{u_j}=0$ for all $i,j$. Adding all $n_0$ vectors from the null space of $C$, we obtain a basis $\{\ket{u_i}\}_{i=1}^{d_C}$ with the property that $\bra{u_i}A\ket{u_j}=0$.

\end{document}

%% file: comandos.tex
\DeclareMathOperator{\tr}{tr}

\DeclareMathOperator{\rank}{rank}

\let\originalleft\left
\let\originalright\right
\renewcommand{\left}{\mathopen{}\mathclose\bgroup\originalleft}
\renewcommand{\right}{\aftergroup\egroup\originalright}

\newcommand{\KetBra}[1]{{\Ket{#1}\!\Bra{#1} }}

\newcommand{\Bra}[1]{{ \langle \! \langle{#1}\vert }}
\newcommand{\Ket}[1]{{ \vert {#1}  \rangle \!  \rangle}}
\newcommand{\bigKet}[1]{{ \Big\vert {#1}  \Big\rangle \!  \Big\rangle}}

\newcommand{\bra}[1]{\langle #1|}

\newcommand{\ket}[1]{| #1 \rangle}

\newcommand{\braket}[2]{\left\langle #1 \middle| #2 \right\rangle}
\newcommand{\ketbra}[1]{\left|#1\middle\rangle\middle\langle#1\right|}

\newcommand{\de}[1]{\left(#1\right)}

\newcommand{\De}[1]{\left[#1\right]}

\newcommand{\mathand}{\quad\text{and}\quad}

\newcommand{\id}{\mathds{1}}

\newcommand{\be}{\begin{equation}}
\newcommand{\ee}{\end{equation}}

\newcommand{\choi}{\mathfrak{C}}

\newcommand{\eg}{e.g.\@\xspace}
\newcommand{\ie}{i.e.\@\xspace}
\newcommand{\etal}{\emph{et al.}\@\xspace}

\newtheorem{theorem}{Theorem}
\newtheorem*{theorem*}{Theorem}

\newtheorem{definition}[theorem]{Definition}

\mathtoolsset{centercolon}